\documentclass[11pt]{article}
\usepackage{authblk}
\usepackage{amssymb,latexsym,amsmath}
\usepackage{amsthm}

\theoremstyle{definition}
\newtheorem{theorem}{Theorem}
\newtheorem{corollary}[theorem]{Corollary}
\newtheorem{proposition}[theorem]{Proposition}
\newtheorem{lemma}[theorem]{Lemma}
\newtheorem{definition}[theorem]{Definition}
\newtheorem{example}[theorem]{Example}
\newtheorem{notation}[theorem]{Notation}
\newtheorem{remark}[theorem]{Remark}
\newtheorem{remarks}[theorem]{Remarks}

\newtheorem{algo}[theorem]{Algorithm}
\usepackage{multirow}
\usepackage{graphicx}
\usepackage{tikz}
\usetikzlibrary{matrix,decorations.pathreplacing}
\usepackage{multirow}

\newcommand{\numberset}{\mathbb}

\newcommand{\F}{\numberset{F}}

\newcommand{\mS}{\mathcal{S}}
\newcommand{\mC}{\mathcal{C}}
\newcommand{\mP}{\mathcal{P}}
\newcommand{\mG}{\mathcal{G}}

\usepackage{hyperref}

\usepackage[margin=3.1cm]{geometry}

\newcommand{\mF}{\mathcal{F}}

\DeclareMathOperator{\spn}{span}

\newcommand{\mM}{\mathcal{M}}

\begin{document}

\title{\textbf{Equidistant subspace codes}}

\author{Elisa Gorla and Alberto Ravagnani\thanks{E-mails:
\texttt{alberto.ravagnani@unine.ch}, \texttt{elisa.gorla@unine.ch}. The authors
were partially supported by the Swiss National Science Foundation through grant no. 200021\_150207 and by the ESF COST Action IC1104.}}
  
 \affil{Institut de Math\'{e}matiques \\ Universit\'{e} de
Neuch\^{a}tel \\ Emile-Argand 11, CH-2000 Neuch\^{a}tel, Switzerland}

\date{}

\makeatletter
\newcommand{\subjclass}[2][2010]{%
  \let\@oldtitle\@title%
  \gdef\@title{\@oldtitle\footnotetext{#1 \emph{Mathematics subject classification:} #2}}%
}
\newcommand{\keywords}[1]{%
  \let\@@oldtitle\@title%
  \gdef\@title{\@@oldtitle\footnotetext{\emph{Keywords:} #1.}}%
}
\makeatother

\subjclass{11T71, 14G50, 94B60, 51E23, 15A21.}

\keywords{network coding, equidistant subspace codes, sunflowers, spreads and partial spreads.}

\maketitle

\begin{abstract}
In this paper we study equidistant subspace codes, i.e. 
subspace codes with the property that each two distinct codewords have the same
distance. We provide an almost complete classification of such codes under the
assumption that the cardinality of the ground field is large enough. More precisely, we prove
that for most values of the parameters, an equidistant code of maximum
cardinality is either a sunflower or the orthogonal of a sunflower. We
also study equidistant codes with extremal parameters, and establish
general properties of equidistant codes that are not sunflowers.
Finally, we propose a systematic construction
of equidistant codes based on our previous construction of partial spread codes,
and provide an efficient decoding algorithm.
\end{abstract}

\section*{Introduction} \label{intr} 

Network coding is a branch of information theory concerned with data transmission over
noisy and lossy networks. A network is modeled by a directed acyclic multigraph,
and information travels from one or multiple sources to multiple receivers through
intermediate nodes.
Network coding has several applications, e.g. peer-to-peer networking,
distributed storage and patches distribution. In \cite{origine} it was proved
that the information rate of a network communication may be improved
employing coding at the nodes of the network, instead of simply routing the
received inputs.
In \cite{origine2} it was shown
that maximal information rate can be achieved in the multicast situation by allowing the intermediate nodes to
perform linear combination of the inputs they receive, provided that the cardinality of the ground field
is sufficiently large. Random linear network coding was introduced in 
\cite{random}, and a mathematical approach
was proposed in \cite{KK1} and \cite{KK2}, together with the definition of
subspace code.

In this paper we study equidistant subspace codes, i.e., subspace codes
with the property that the intersection of any pair of codewords has the same dimension.
Equidistant subspace codes were shown to have relevant applications 
in distributed storage in~\cite{nat}.
In the same paper, Etzion and Raviv identify two trivial families of equidistant codes, namely
sunflowers and balls. A ball is a subspace code in the Grassmannian $\mG_q(k,n)$ of 
$k$-dimensional subspaces of $\F_q^n$ with the property that all the elements of the code are 
cointained in a fixed $(k+1)$-dimensional subspace of $\F_q^n$. They proceed then to study 
the question of when an equidistant code belongs to one of the two families. 
Starting from the observation that the orthogonal of a ball is a sunflower, in this paper we study 
the question of when an equidistant code is either a sunflower or the orthogonal of a sunflower. 
One of our main results is a classification of equidistant subspace codes over fields of 
large enough cardinality: We prove that, for most choices of the parameters, an
equidistant code of maximum cardinality is either a sunflower or 
the orthogonal of a sunflower. In addition, for most values of the parameters 
the two possibilities are mutually exclusive. 
We also study extremal equidistant codes, i.e. codes for which every two distinct 
codewords intersect in codimension one. 
We show that each such code is either a sunflower of the orthogonal 
of a sunflower, over fields of any size and for a code of any cardinality.
We also establish general properties of equidistant codes that are not sunflowers.
Finally, we give a systematic construction of asymptotically optimal
equidistant codes based on the construction of partial spread codes 
from~\cite{noi}. We then exploit the structure of our codes to design an
efficient decoding algorithm for them and for their orthogonals.

The paper is organized as follows: In Section \ref{prel} we recall some definitions and results on subspace codes, 
equidistant codes, sunflowers and partial spreads. In Section \ref{specparam} we study extremal equidistant codes, 
with the property that each two distinct elements intersect in codimension one.
In Section \ref{propr} we give a classification of equidistant codes for most values 
of $k,n$ and for $q\gg 0$. The classification is summarized in Theorem~\ref{class}.
In Section \ref{other} we study equidistant
codes that are not sunflowers. In Section~\ref{con} we give a systematic
construction for sunflower codes, and we argue that their cardinality is
asymptotically optimal.
In Section \ref{de} we show how to decode them efficiently and in Section
\ref{duals} we explicitely describe their orthogonal codes and show how to decode them.

\section{Preliminaries} \label{prel}

We briefly recall the main definitions and results on subspace codes, equidistant
codes, and partial spreads.

\begin{notation}
Throughout the paper $q$ denotes a fixed prime power, and $k,n$ two 
integers with $1 \le k < n$. We denote by $\mG_q(k,n)$ the set of
$k$-dimensional vector subspaces of $\F_q^n$.
\end{notation}

\begin{definition} \label{maindef}
 The \textbf{subspace distance} between subspaces $U,V \subseteq \F_q^n$
 is defined by $$d(U,V):=\dim(U)+\dim(V)-2\dim(U \cap V).$$ A \textbf{subspace code}
of constant dimension $k$
 is a subset $\mC \subseteq \mG_q(k,n)$ with
 $|\mC|\ge 2$. The \textbf{minimum distance} of $\mC$ is
 $d(\mC):=\min \{ d(U,V) : U,V \in \mC, U \neq V\}$.
 The code $\mC$ is \textbf{equidistant} if for all
 $U,V \in \mC$ with $U \neq V$ we have $d(U,V)=d(\mC)$. An equidistant code
$\mC \subseteq \mG_q(k,n)$ is \textbf{$c$-intersecting} if $d(\mC)=2(k-c)$.
\end{definition}

Notice that equidistant $c$-intersecting codes exist only for $n\geq 2k-c$, 
since codes contain at least two codewords. 
%constant dimension subspace codes have even minimum distance.

\begin{notation}
Given an integer $0 \le c \le k-1$, we denote by $e_q(k,n,c)$ the largest
cardinality of an equidistant $c$-intersecting subspace code $\mC
\subseteq \mG_q(k,n)$. 
\end{notation}

\begin{definition}
An equidistant $c$-intersecting code $\mC \subseteq \mG_q(k,n)$ is {\bf optimal} if $|\mC|=e_q(k,n,c)$.
A family of codes $\mC_q \subseteq \mG_q(k,n)$ is {\bf asymptotically optimal} if 
$\lim_{q\to\infty}|\mC_q|/e_q(k,n,c)=1$.
\end{definition}
%
%Minimum distance and error and erasure correction capability relate as follows. 
%See \cite{KK1} for the definition of error and erasure in the context
%of network coding.
%
%\begin{theorem}[\cite{KK1}, Theorem 2]
%A constant dimension code $\mC \subseteq \mG_q(k,n)$ of minimum distance
%$d$ corrects $t$ errors and $e$ erasures if and only if $2(t+k-e)<d$. 
%\end{theorem}
%
%Partial spreads are equidistant codes with minimum distance equal to $2k$,
%i.e., $0$-intersecting equidistant codes.
Partial spreads are a first example of equidistant subspace codes.

\begin{definition}
 A \textbf{partial spread} in $\mG_q(k,n)$ is a subspace code $\mS \subseteq \mG_q(k,n)$
 with $d(\mS)=2k$.
\end{definition}
The maximum cardinality of a partial spread
$\mS \subseteq \mG_q(k,n)$ is $e_q(k,n,0)$ by definition. 
A systematic construction for partial spreads and an efficient decoding algorithm are 
given in \cite{noi}. The cardinality of the codes from~\cite{noi} meets the 
lower bound of the following well-known result (see e.g. \cite{Beu_1}). 
It follows that the codes are asymptotically optimal.

\begin{theorem}\label{bo}
 Let $r$ denote the remainder obtained dividing $n$ by $k$. We have 
 $$\frac{q^n-q^r}{q^k-1}-q^r+1 \le e_q(k,n,0) \le \frac{q^n-q^r}{q^k-1}.$$
\end{theorem}

\begin{remark} \label{pss}
The lower and upper bound of Theorem \ref{bo} agree when $r=0$. 
In this case $k$ divides $n$ and the bound is always attained by codes that
are called \textbf{spreads} (see \cite{GR} and the references within). 
When $k$ does not divide $n$, deciding whether $e_q(k,n,0)$ may be equal to the upper bound for 
some values of $q,k,n$ is an open problem. 
For some special values of $q,k,n$ moreover, the lower bound of Theorem~\ref{bo} can be improved,
see e.g. \cite{Beu_2} and \cite{Drake_etc}.
\end{remark}

Sunflowers are a main source of examples of equidistant codes.

\begin{definition}
 A subspace code $\mF \subseteq \mG_q(k,n)$ is a \textbf{sunflower} if there exists a subspace $C \subset \F_q^n$
 such that for all $U,V \in \mF$ with $U \neq V$ we have $U\cap V=C$.
 The space $C$ is called the \textbf{center} of the sunflower $\mF$.
\end{definition}

 A sunflower $\mF \subseteq \mG_q(k,n)$ with center $C$ of dimension
$c$ is an equidistant $c$-intersecting
subspace code
 with minimum distance $2(k-c)$. 
%By Definition \ref{maindef} we have $0 \le c \le k-1$.
The connection between partial spreads and sunflowers is described in the
following simple remark. The same
observation appears in \cite{etzion1}, Theorems 10 and 11.

\begin{remark} \label{legame}
 Let $\mF \subseteq \mG_q(k,n)$ be a sunflower with center $C$
of dimension $c$, and let $\varphi: \F_q^n/C \to \F_q^{n-c}$ be an
isomorphism.
 Then the subspace code $$\mS:=\{ \varphi(U/C) : U \in \mF\} \subseteq
\mG_q(k-c,n-c)$$
 is a partial spread with $|\mS|=|\mF|$. Conversely, given an
integer $0 \le c \le k-1$, a partial spread
 $\mS \subseteq \mG_q(k-c,n-c)$, and a subspace $C \subseteq \F_q^n$, the subspace code
$\mF:= \{ C \oplus U : U \in \mS \} \subseteq \mG_q(k,n)$ is a sunflower with
center $C$ and $|\mF|=|\mS|$.
\end{remark}

By Remark \ref{legame} one easily obtains the

\begin{corollary}\label{cor} For all $0 \le c \le k-1$ we have $e_q(k,n,c)
\ge e_q(k-c,n-c,0)$.
\end{corollary}

The following result shows that equidistant codes of large cardinality are sunflowers.
The proof is based on a result by Deza on classical codes (see \cite{dfO} and \cite{df}), applied in
the context of network coding by Etzion and Raviv.

\begin{theorem}[\cite{etzion1}, Theorem 1] \label{sun}
 Let $0 \le c \le k-1$ be an integer, and let $\mC \subseteq \mG_q(k,n)$ be a
 $c$-intersecting equidistant code. Assume that $$|\mC| \ge ((q^k-q^c)/(q-1))^2+
(q^k-q^c)/(q-1)+1.$$
 Then $\mC$ is a sunflower.
\end{theorem}

\begin{remark}
 Deza conjectured that any $c$-intersecting equidistant code $\mC \subseteq
\mG_q(k,n)$ with $|\mC|>(q^{k+1}-1)/(q-1)$
is a sunflower (see \cite{etzion1}, Conjecture 1). The conjecture was
disproved in \cite{etzion1}, Section 3.2, where the authors give an example of an
equidistant code $\mC \subseteq \mG_2(3,6)$ of minimum distance $4$ and 
cardinality $16$, which is not a sunflower. The example was found by computer search.
\end{remark}

We close this section with the definition of orthogonal and span of a code.

\begin{definition}
 The \textbf{orthogonal} of a code $\mC \subseteq \mG_q(k,n)$ is 
 $\mC^\perp:= \{ U^\perp : U \in \mC\} \subseteq \mG_q(n-k,n)$, where $U^\perp$
is the orthogonal of $U$ with respect to the standard inner product of $\F_q^n$.
\end{definition}

\begin{remark} \label{dua}
For any $U,V \in \mG_q(k,n)$ we have 
$\dim(U^\perp \cap V^\perp)=n-2k+\dim(U\cap V)$. Hence $d(\mC)=d(\mC^\perp)$ 
for any subspace code $\mC \subseteq \mG_q(k,n)$.
In particular, the orthogonal of a $c$-intersecting equidistant code $\mC
\subseteq
\mG_q(k,n)$ is a $(n-2k+c)$-intersecting equidistant code (see also Theorem 13 and Theorem 14 of \cite{etzion1}).
 Notice that $n-2k+c\geq 0$, 
since $\mC$ contains two distinct codewords.
This proves that $$e_q(k,n,c)=e_q(n-k,n,n-2k+c)$$ for all $0 \le c \le k-1$,
and the orthogonal of an optimal equidistant code is an optimal equidistant code.
\end{remark}

\begin{definition}
Let $\mC \subseteq \mG_q(k,n)$ be a subspace code. We define the \textbf{span}
of $\mC$ as 
$$\spn(\mC):= \sum_{U \in \mC} U \subseteq
\F_q^n.$$ 
\end{definition}

\section{Extremal equidistant codes} \label{specparam}

In this section we study $(k-1)$-intersecting codes in $\mG_q(k,n)$. 
We call such codes {\bf extremal}, since $k-1$ is the largest possible value of $c$, for 
given $k$ and $n$. 
%Since our codes contain at least two codewords, one has $n-2k+c=n-k-1\geq 0$, hence $n\geq k+1$.
%Notice that the orthogonal of an extremal code $\mC\subset\mG_q(k,n)$ is again extremal, since 
%$\mC^\perp\subset\mG_q(n-k,n)$ is $(n-k-1)$-intersecting by Remark \ref{dua}.
Notice that these codes are equidistant with minimum distance $d=2$. 
In particular, the orthogonal of an extremal code is extremal.
Our main result shows that every extremal equidistant code is either a sunflower, 
or the orthogonal of a sunflower.
In Section \ref{propr} we establish a similar result for most choices of $(k,n,c)$ and for $q \gg 0$.

\begin{proposition}\label{ccequi}
 Let $\mC \subseteq \mG_q(k,n)$ be a $c$-intersecting equidistant code. The
following are equivalent:
\begin{enumerate}
 \item $\mC$ is a sunflower,
 \item $\dim \mbox{span}(\mC^{\perp})=n-c$,
 \item for all $A,B \in \mC$ with $A \neq B$ we have $\mbox{span}(\mC^\perp)=A^\perp+B^\perp$.
\end{enumerate}
\end{proposition}

\begin{proof}
Properties $(2)$ and $(3)$ are clearly equivalent. The code $\mC$ is a
sunflower if and only if there exists $C \subseteq \F_q^n$ with $\dim(C)=c$
such that $A\cap B = C$ for all $A,B \in \mC$ with $A \neq B$. The condition
$A\cap B=C$ is equivalent to $A^\perp+B^\perp=C^\perp$. Hence $(1)$ and $(3)$
are equivalent.
\end{proof}

The following may be regarded as classification of $(k-1)$-intersecting codes in $\mG_q(k,n)$.

\begin{proposition} \label{c=k-1}
 Let $\mC \subseteq \mG_q(k,n)$ be a $(k-1)$-intersecting equidistant code.
 Then either $\mC$ is a sunflower, or $\mC^\perp$ is a sunflower. 
\end{proposition}

\begin{proof}
 If $|\mC|=2$ the result is trivial. Assume $|\mC| \ge 3$ and that $\mC$ is not
a sunflower. Let $A,B \in \mC$ with $A \neq B$. By
Proposition \ref{ccequi} it suffices to show that $\mbox{span}(\mC)=A+B$. Since
$\mC$ is not a sunflower, there exists $D \in \mC \setminus \{ A,B\}$
such that $D\cap A\neq D\cap B$. Since
$D\supseteq D \cap A + D\cap B$ and $\dim(D \cap A + D\cap B) \ge \dim(D \cap A)+1=k$, 
then $D=D \cap A + D\cap B \subseteq A+B$. For any $E\in \mC \setminus \{
A,B,D\}$ we have
$E \supseteq E\cap A + E\cap B + E\cap D$. Since $\dim (A\cap B\cap D) <k-1$, then $E \cap A$, $E\cap B$, and $E
\cap D$ are not all equal. Hence $\dim(E\cap A + E \cap B + E \cap D)\geq \dim(E\cap A)+1=k$ 
and $E=E\cap A + E \cap B + E \cap D\subseteq A+B$. Therefore $\mbox{span}(\mC)=A+B$. 
\end{proof}

As a corollary, we obtain an improvement of Theorem 12 and Corollary 1 of \cite{etzion1}.

\begin{corollary}\label{extr_bound}
Let $\mC \subseteq \mG_q(k,n)$ be a $(k-1)$-intersecting equidistant code. 
If $\mC$ is not a sunflower, then it is a subset of the set of $k$-dimensional subspaces
of a given $(k+1)$-dimensional space.
In particular, if $|\mC| > \begin{bmatrix} k+1
\\ k \end{bmatrix}$, then $\mC$ is a sunflower. 
\end{corollary}

\begin{proof}
 If $\mC$ is not a sunflower, then by Proposition \ref{c=k-1} the orthogonal code
$\mC^\perp$ is a sunflower. By Proposition \ref{ccequi} this is equivalent to
$\dim \mbox{span}(\mC)=k+1$. Then all the codewords of
$\mC$ are contained in a fixed $(k+1)$-dimensional space of $\F_q^n$. In particular,
their number cannot exceed the number of $k$-dimensional subspaces of a
$(k+1)$-dimensional space, a contradiction. 
\end{proof}

\begin{remarks}
\begin{enumerate}
\item Combining Theorem \ref{bo}, Corollary \ref{cor}, and Corollary~\ref{extr_bound}, 
we have that if $\mC$ has maximum cardinality $e_q(k,n,k-1)$ and $n \gg 0$, then $\mC$ is a sunflower.
\item As observed in \cite{etzion1}, the bound of Corollary~\ref{extr_bound} is sharp for any $k,n$. 
In fact, let $\mC$ be the set of 
$k$-dimensional subspaces of a fixed $(k+1)$-dimensional space of $\mathbb{F}_q^n$.
$\mC$ is an equidimensional $(k-1)$-intersecting code of cardinality $\begin{bmatrix} k+1
\\ k \end{bmatrix}$ which is not a sunflower. 
%Notice that $\mC^\perp$ is a sunflower by Proposition~\ref{c=k-1}.
\end{enumerate}
\end{remarks}

\section{A classification of equidistant codes} \label{propr}

In this section we provide a classification of optimal equidistant codes for most 
values of the parameters. 
More precisely we prove that, for $q\gg0$ and for most values of $k$ and $n$,  
every optimal equidistant code is either a sunflower or the orthogonal of a sunflower.
%Our classification heavily relies on the results of Deza (\cite{dfO} and \cite{df}) 
%and of Etzion-Raviv (\cite{etzion1}) discussed in the previous section. 
We start by studying the case when $k$ is small with respect to $n$.

\begin{proposition}\label{noi}
Let $q\gg 0$ and $n\geq 3k-1$. Then $$e_q(k,n,c)=e_q(k-c,n-c,0).$$
Moreover, any $c$-intersecting equidistant code $\mC\subseteq \mG_q(k,n)$ of
cardinality $e_q(k,n,c)$ is a sunflower.
\end{proposition}

\begin{proof}
 Let $0 \le r \le k-c-1$ denote the remainder obtained dividing $n-c$ by $k-c$. 
Since $n>3k-2 \ge 2k-1$, we have $r \le k-1<n-k$. Therefore
$$\lim_{q\rightarrow\infty}\frac{\frac{q^{n-c}-q^r}{q^{k-c}-1}-q^r+1}{q^{n-k}}=1.$$
On the other hand $$\lim_{q \to \infty} \frac{\left(\frac{q^k-q^c}{q-1}\right)^2+ \frac{q^k-q^c}{q-1}+1}{q^{2k-2}}=1.$$
Since $k<(n+2)/3$ we have $n-k>2k-2$. Hence
$$\lim_{q \to \infty}\frac{\frac{q^{n-c}-q^r}{q^{k-c}-1}-q^r+1 - 
\left[\left(\frac{q^k-q^c}{q-1}\right)^2+ \frac{q^k-q^c}{q-1}+1\right]}{q^{n-k}}=1.$$
In particular, for $q\gg 0$ we have
$$\frac{q^{n-c}-q^r}{q^{k-c}-1}-q^r+1 \ge \left(\frac{q^k-q^c}{q-1}\right)^2+ 
\frac{q^k-q^c}{q-1}+1.$$ 
By Theorem \ref{bo} and Corollary \ref{cor} we have
$$|\mC| \ge \frac{q^{n-c}-q^r}{q^{k-c}-1}-q^r+1 \ge \left(\frac{q^k-q^c}{q-1}\right)^2+ 
\frac{q^k-q^c}{q-1}+1.$$ Theorem \ref{sun} implies that $\mC$ is a sunflower. 
Hence by  Remark \ref{legame} we have $e_q(k,n,c) = e_q(k-c,n-c,0)$.
\end{proof}

For completeness we also examine the case when $n$ is small with respect to $k$.

\begin{proposition} \label{primo}
Let $q \gg 0$ and $n\leq (3k+1)/2$. Then $$e_q(k,n,c) = e_q(k-c,2k-c,0).$$
Moreover, every $c$-intersecting equidistant code $\mC \subseteq \mG_q(k,n)$ of
cardinality $e_q(k,n,c)$ is of the form 
$\mS^\perp$, where $\mS$ is a sunflower.
\end{proposition}

\begin{proof}
By Remark \ref{dua} we have $e_q(k,n,c)=e_q(n-k,n,n-2k+c)$. Since $n \ge 3(n-k)-1$, the thesis follows from
Proposition \ref{noi}.
\end{proof}

Proposition \ref{noi} and Proposition \ref{primo} 
imply that for $n \le (3k+1)/2$ or for $n \ge 3k-1$ and $q\gg 0$, 
every equidistant code of maximum cardinality $e_q(k,n,c)$ is either a
sunflower, or the orthogonal of a sunflower. 
We now show that these families are almost always disjoint.

\begin{lemma} \label{not}
Let $\mS \subseteq \mG_q(k,n)$ be a sunflower with center $C$ of dimension $0
\le c \le k-1$ and $\mbox{span}(\mS)=\F_q^n$. Assume that $n>2k-c$. 
Then $\mS^\perp$ is not a sunflower.
\end{lemma}

\begin{proof}
By contradiction, assume that $\mS^\perp \subseteq\mG_q(n-k,n)$ is a sunflower
with center $D$. By
Remark \ref{dua} $\dim(D)=n-2k+c>0$. Moreover, $D \subseteq
U^\perp$ for all 
$U \in \mS$, i.e., 
$U \subseteq D^\perp$ for all $U \in \mS$. Then $\F_q^n=\spn(\mS)
\subseteq D^\perp$, 
which contradicts 
the assumption that $D\neq 0$. 
\end{proof}

Remark~\ref{dua} and Lemma~\ref{not}
allow us to construct a family of equidistant 
codes which are not sunflowers and have maximum cardinality for their
parameters.

\begin{example}\label{spreaddual}
Let $n=\ell k$, $\ell>2$. Let $\mS\subseteq\mG_q(k,\ell k)$ be a spread. 
Then $\mS^\perp$ is an optimal equidistant code which is not a sunflower by 
Lemma~\ref{not}. We have $$|\mS^\perp|=|\mS|=e_q(k,\ell k,0)=e_q((\ell-1)k,\ell k,(\ell-2)k),$$ 
where the last equality follows from Remark~\ref{dua}. 

Setting $k=1$ we recover two well-known examples of equidistant codes: 
$\mS$ is the set of lines in $\F_q^\ell$ and $\mS^\perp$ is the set of $(\ell-1)$-dimensional 
subspaces of $\F_q^\ell$. 
%For $k>1$, $\mS$ is a $k$-spread in $\F_q^{k\ell}$, and $\mS^\perp$ is an
%optimal set
%of $(k\ell-k)$-spaces in $\F_q^{k \ell}$ having pairwise intersections of
%dimension
%$k \ell-2k$ which is not a sunflower. 
\end{example}

Now we prove that a $c$-intersecting sunflower $\mS \subseteq \mG_q(k,n)$ with maximum cardinality $e_q(k,n,c)$ is never contained
in a proper subspace of $\F_q^n$.

\begin{proposition}\label{full}
 Let $\mS \subseteq \mG_q(k,n)$ be a sunflower with center of dimension $0 \le
c \le k-1$. Let $r$ denote the remainder obtained dividing $n-c$ by $k-c$. If
$$|\mS| \ge \frac{q^{n-c}-q^r}{q^{k-c}-1}-q^r+1,$$ then $\mbox{span}(\mS)=\F_q^n$.
In particular, if $|\mS| = e_q(k,n,c)$ then $\mbox{span}(\mS)=\F_q^n$.
\end{proposition}

\begin{proof}
Since $\mS$ is a sunflower with center of dimension $c$, we have
\begin{eqnarray*}
 |\bigcup_{V \in \mS} V|&=& q^c+ |\mS| (q^k-q^c) \\
 &\ge& q^c+q^c(q^{k-c}-1) \left( \frac{q^{n-c}-q^r}{q^{k-c}-1}-q^r+1  \right)
\\ &=& q^n+q^k-q^{k+r} \\
&\ge& q^n+q^k-q^{2k-c-1}.
\end{eqnarray*}
Since $|\mS| \ge 2$, then $n \ge 2k-c$, hence
$q^n+q^k-q^{2k-c-1} \ge q^n+q^k-q^{n-1} > q^{n-1}$. Therefore $\mS$ cannot be contained in a proper subspace of 
$\F_q^n$. The second part of the statement follows from Corollary \ref{cor} and Theorem \ref{bo}.
\end{proof}

\begin{corollary} \label{cc}
 Let $\mC \subseteq \mG_q(k,n)$ be a $c$-intersecting equidistant code with
$|\mC|=e_q(k,n,c)$. Then $\mC$ and
$\mC^\perp$
are both sunflowers if and only if $n=2k$ and both $\mC$ and $\mC^\perp$ are
spreads.
\end{corollary}

\begin{proof}
 Assume that both $\mC$ and $\mC^\perp$ are sunflowers. 
By Remark \ref{dua} the center of $\mC^\perp$ has dimension
$c'=n-2k+c\ge 0$. Since $|\mC|=e_q(k,n,c)$,
Proposition \ref{full} and Lemma \ref{not} applied to
$\mC$ give $c'=0$. In particular, $\mC^\perp$ is a partial spread.
Since $\mC$ is optimal, then $\mC^\perp$ is optimal by Remark~\ref{dua}.
By Proposition \ref{full} and Lemma \ref{not}, $c=n-2(n-k)+c'=0$. 
Hence $n=2k$ and $\mC$ is a partial spread. Since
$n=2k$ and $\mC$ and $\mC^\perp$ have maximum cardinality, they are spreads.
\end{proof}

Hence when $n=2k$ and $c=0$, every $0$-intersecting equidistant code $\mC \subseteq
\mG_q(k,2k)$ of maximum cardinality is a spread, and its orthogonal is again a spread with
the same parameters. Therefore every equidistant code of maximum cardinality is a 
sunflower, as well as its orthogonal.

%\begin{remark} \label{oss1}
For $n=2k$ and $c>0$, $e_q(k,2k,c) \geq e_q(k-c,2k-c,0)$ by
Corollary~\ref{cor}, and the 
two quantities do not always agree, e.g. $$e_q(3,6,1)>e_q(2,5,0),$$
as shown in the 
next example.
Moreover, for any $k,c$ for which $e_q(k,2k,c)=e_q(k-c,2k-c,0)$,
let 
$\mC\subseteq\mG_q(k,2k)$ be a $c$-intersecting sunflower of cardinality
$e_q(k,2k,c)$. Then by 
Corollary~\ref{cc} we also have a $c$-intersecting equidistant code
$\mC^\perp\subseteq\mG_q(k,2k)$
of maximum cardinality  
which is not a sunflower. Hence for any $k,c$ we have $c$-intersecting 
equidistant codes $\mC\subseteq\mathcal{G}(k,2k)$ of maximum cardinality which are not sunflowers, 
but we may not always have sunflower codes of the same cardinality.

In addition, it may be possible to also 
have an equidistant code $\mC\subseteq\mG_q(k,2k)$ of cardinality $e_q(k,2k,c)$
such 
that neither $\mC$ nor $\mC^\perp$ are sunflowers. This is the case of the
following example.
%\end{remark}

\begin{example}[\cite{bem}, Example~1.2] \label{exkl}
The hyperbolic Klein set $\mC\subseteq\mG(3,6)$ is an equidistant code with
$c=1$ and $|\mC|=q^3+q^2+q+1$. 
$\mC$ is not a sunflower, nor the orthogonal of a sunflower, since the largest
possible cardinality of a sunflower with $k=3, n=6, c=1$ is 
$$e_q(2,5,0)\leq \frac{q^5-q}{q^2-1}=q^3+q<|\mC|=|\mC^\perp|,$$
where the inequality follows from Theorem~\ref{bo}. In particular, 
$e_q(3,6,1)>e_q(2,5,0).$
\end{example}

%\begin{remark}
% The hyperbolic Klein set of Example \ref{exkl} is a well-known object in
%discrete mathematics. Given an hyperbolic quartic $\mathcal{Q}$ in a
%5-dimensional projective space over $\F_q$, one can consider a special family
%of planes, called $\mathcal{Q}$-planes, arising from $\mathcal{Q}$. It can
%be shown that such family of $\mathcal{Q}$-planes splits into two
%equivalent classes, one of which is precisely the hyperbolic Klein set of
%Example \ref{exkl}. We refer the interested reader to \cite{bebook}, Chapter 4. 
%\end{remark}

Combining Propositions  \ref{c=k-1}, \ref{noi}, \ref{primo}, \ref{full},  and 
Corollary \ref{cc} one
easily obtains the following classification of equidistant codes of maximum cardinality.

\begin{theorem} \label{class}
Let $\mC \subseteq \mG_q(k,n)$ be a $c$-intersecting equidistant code with $|\mC|=e_q(k,n,c)$. 
Assume that one of the following conditions holds:
\begin{itemize}
\item $c\in\{0,k-1,2k-n\}$,
\item $n \le (3k+1)/2$ and $q \gg 0$,
\item $n \ge 3k+1$ and $q \gg 0$.
\end{itemize}
Then either $\mC$ is a sunflower or $\mC^\perp$ is a sunflower, 
and the two are mutually exclusive unless $c=0$ and $n=2k$.
\end{theorem}

%Theorem \ref{class} provides a classification of equidistant codes
%of maximum cardinality for almost all parameters and $q$ sufficiently
%large. 
Notice that $n\gg k$ is the relevant practical situation within network  coding. 
Moreover, one needs to assume $q\gg 0$ in order  to have a solution to the
network coding problem
(see \cite{med}, Chapter 1 for details).

\section{Other properties of equidistant codes} \label{other}

We devote this section to equidistant codes that are not sunflowers. 
The property of having a center characterizes sunflowers among equidistant codes.

\begin{definition}
Let $\mC \subseteq \mG_q(k,n)$ be a $c$-intersecting equidistant code,
$0 \le c \le k-1$. The \textbf{set of centers} of $\mC$ is
$T(\mC):= \{ U \cap V : U,V \in \mC, \ U \neq V\}$, and the \textbf{number of
centers} of $\mC$ is $t(\mC):=|T(\mC)|$. The \textbf{set of petals} attached to
a  center
$A \in T(\mC)$ is $\mP(A):= \{ U \in \mC : A \subseteq U\}$.
\end{definition}

In the next proposition we show that equidistant codes that
have many codewords are either sunflowers, or they have a large number of
centers.

\begin{proposition}\label{last}
Let $\mC \subseteq \mG_q(k,n)$ be an $c$-intersecting equidistant code,
$0 \le c \le k-1$. 
One of the following  properties holds:

\begin{enumerate}
\item $\mC$ is a sunflower, or
\item $t(\mC) \ge |\mC|\frac{q^c-q^{c-1}}{q^k-q^{c-1}}$.
\end{enumerate}
\end{proposition}

\begin{proof}
If $\mC$ is not a sunflower, then $t:=t(\mC) \ge 2$. 
Choose an enumeration $T(\mC)=\{ A_1,...,A_t\}$. Since $\mC= \bigcup_{i=1}^t \mP(A_i)$, we have 
\begin{equation} \label{in}  
|\mC| \le \sum_{i=1}^t |\mP(A_i)|.
\end{equation}
For any $i \in \{ 1,...,t\}$, $\mP(A_i)$ is a sunflower with $c$-dimensional center $A_i$, minimum distance $2(k-c)$, and cardinality $s_i:=|\mP(A_i)|$.
If $V\in\mC\setminus\mP(A_i)$, then 
$$|V|\geq\left|V\cap\bigcup_{U\in\mP(A_i)}U\right|=\sum_{U\in\mP(A_i)}|V\cap U|-(s_i-1)|V\cap A_i|=s_i|A_i|-(s_i-1)|V\cap A_i|$$
hence $q^k\geq s_iq^c-(s_i-1)q^{c-1}=s_i(q^c-q^{c-1})+q^{c-1}.$
Therefore we have shown that $$|\mP(A_i)|\leq\frac{q^k-q^{c-1}}{q^c-q^{c-1}}$$
for all $1\leq i\leq t$, and the thesis follows by (\ref{in}).
\end{proof}

In particular, for a code with maximum cardinality which is not a sunflower, we can 
give the following asymptotic estimate of the number of centers as $q$ grows.

\begin{corollary}\label{estimate}
Let $\mC\subseteq\mG_q(k,n)$ be a $c$-intersecting equidistant code. 
Assume that $|\mC|=e_q(k,n,c)$ and that $\mC$ is not a sunflower. Denote by $r$
the remainder of 
the division of $n-c$ by $k-c$. Then
$$t(\mC)\geq e_q(k,n,c)\frac{q^c-q^{c-1}}{q^k-q^{c-1}}\geq 
\left(\frac{q^{n-c}-q^r}{q^{k-c}-1}-q^r+1\right)\frac{q^c-q^{c-1}}{q^k-q^{c-1}}.$$ 
In particular, $\lim_{q\to\infty} t(\mC)q^{-(n-2k+c)}\in [1,+\infty]$.
\end{corollary}

\begin{proof}
The inequality follows by Proposition~\ref{last}, Corollary~\ref{cor}, and Theorem~\ref{bo}. 
Hence $$\lim_{q\to\infty}t(\mC)q^{-(n-2k+c)}\geq\lim_{q\to\infty}
\left(\frac{q^{n-c}-q^r}{q^{k-c}-1}-q^r+1\right)
\frac{q^c-q^{c-1}}{q^k-q^{c-1}}q^{-(n-2k+c)}=1,$$
as claimed.
\end{proof}

%\begin{remark}
% Notice that the two properties of Proposition \ref{last} are not mutually
%exclusive in general. Let e.g. $n=3$ and take two spaces $A,B
%\subseteq \F_q^n$ of dimension $k=2$ that intersect in dimension $c=1$. 
%Then $\mC:= \{ A,B\}$ is a sunflower and $t(\mC)=1 \ge 2/(q+1)$.
%\end{remark}

The orthogonal of a sunflower is often an example of an optimal code with a large number of centers.

\begin{example}
Let $0< c\leq k-1$, $\mS\subset\mG_q(n-k,n)$ be a sunflower of maximum cardinality with 
$(n-2k+c)$-dimensional center. Let $\mC=\mS^\perp\subset\mG_q(k,n)$, then $\mC$ is 
$c$-intersecting and $|\mC|=|\mS|$. $\mC$ is not a sunflower by Corollary~\ref{cc} and it has 
$$t(\mC)={|\mC| \choose 2}.$$
In fact, for any $A,B,D\in\mS$ pairwise distinct one has $$\dim(A+B)^\perp=n-2k+c>n-3k+2c=\dim(A+B+D)^\perp,$$ 
hence $$A^\perp\cap B^\perp\neq A^\perp\cap B^\perp\cap D^\perp.$$ In particular, there exist no distinct 
$A^\perp,B^\perp,D^\perp\in\mC$ such that $A^\perp\cap D^\perp=B^\perp\cap D^\perp$. 
Similarly one shows that there exist no distinct $A^\perp,B^\perp,D^\perp, E^\perp\in\mC$ 
such that $A^\perp\cap D^\perp=B^\perp\cap E^\perp$.
\end{example}

\section{A systematic construction of sunflower codes} \label{con}

In this section we modify the construction of partial spreads proposed in \cite{noi} to systematically produce sunflower codes for any choice of $k,n,c$. We are motivated by Proposition~\ref{noi} where we show that every equidistant code of maximum cardinality is a sunflower, provided that $q\gg 0$ and $n\geq 3k-1$. An efficient decoding algorithm is given in Section \ref{de}.

\begin{notation}
Denote by $I_m$ an identity matrix of size $m\times m$, by $0_m$ a zero matrix of size $m\times m$, and by $0_{m\times\ell}$ a zero matrix of size $m\times\ell$.
\end{notation}

\begin{definition} \label{companion}
Let $p \in \F_q[x]$ be an irreducible monic polynomial of degree $s \ge 1$.
Write $p(x)=x^s+\sum_{i=0}^{s-1} p_ix^i$. The \textbf{companion matrix}
of $p$ is the $s \times s$ matrix
$$\mbox{M}(p):=\begin{bmatrix}
   0 & 1 & 0 & \cdots & 0 \\
0 & 0 & 1 &  & 0 \\
\vdots & & & \ddots & \vdots \\
0 & 0 & 0 &  & 1 \\
-p_0 & -p_1 & -p_2 & \cdots &-p_{s-1} 
  \end{bmatrix}.$$ 
\end{definition}

The construction of sunflower codes which  we propose is based on companion
matrices of polynomials. It extends the constructions of~\cite{GR} and~\cite{noi}.

\begin{theorem} \label{nsconstr}
 Let $1\leq k<n$ and $\min\{0,2k-n\}\le c \le k-1$ be integers. Write
$n-c=h(k-c)+r$, with $0 \le r \le k-c-1$, $h \ge 2$.
Choose irreducible monic polynomials $p,p' \in \F_q[x]$
of degree $k-c$ and $k-c+r$, respectively. Set $P:=\mbox{\textnormal{M}}(p)$
and $P':=\mbox{\textnormal{M}}(p')$.
For $1 \le i \le h-1$ let
$\mM_i(p,p')$ be the set of $k \times n$ matrices of the form
$$\begin{bmatrix}
I_c & 0_{c\times (k-c)} & \cdots & \cdots & \cdots & \cdots & \cdots & 0_{c\times (k-c)} & 0_{c\times (k-c+r)} \\
0_{(k-c) \times c} & 0_{k-c} & \cdots & 0_{k-c} & I_{k-c} & A_{i+1} & \cdots & A_{h-1} & A_{[k-c]}
\end{bmatrix},$$
where we have $i-2$ consecutive copies of $0_{k-c}$, the matrices $A_{i+1},...,A_{h-1} \in \F_q[P]$, $A \in \F_q[P']$, and $A_{[k-c]}$
denotes the last $k-c$ rows of $A$.
The set 
\begin{eqnarray*}
 \mC := &\bigcup_{i=1}^{h-1}& \{ \mbox{rowsp}(M) : M \in \mM_i(p,p')\} \\
&\cup& \left\{ \mbox{rowsp} \begin{bmatrix}
I_c & 0_{c\times (k-c)} & \cdots & 0_{c \times (k-c)} & 0_{c \times (k-c+r)} & 0_{c \times (k-c)} \\
0_{(k-c) \times c} & 0_{k-c} & \cdots & 0_{k-c} & 0_{(k-c) \times (k-c+r)} & I_{k-c}
\end{bmatrix}\right\}
\end{eqnarray*}
is a sunflower in $\mG_q(k,n)$ of cardinality
$$|\mC|=\frac{q^{n-c}-q^r}{q^{k-c}-1}-q^r+1.$$
\end{theorem}

\begin{proof}
Let $C:= \{ v \in \F_q^n : v_i=0 \mbox{ for } i > c\}$. To simplify the
notation, let $B$ denote the matrix
$$\begin{bmatrix}
              I_c & 0_{c\times k-c} & \cdots & 0_{c \times k-c} & 0_{c \times
k-c+r}
& 0_{c \times k-c} \\
0_{k-c \times c} & 0_{k-c} & \cdots & 0_{k-c} & 0_{k-c \times k-c+r} & I_{k-c}
             \end{bmatrix}.$$
Given a matrix $M \in \mM_i(p,p') \cup \{ B\}$, 
let $\overline{M}$ be
the
matrix obtained from $M$ by deleting the first $c$ rows.
We identify $\F_q^{n-c}$ with $\{ v \in \F_q^n : v_i=0 \mbox{ for
} i=1,...,c\}$,
so that $\F_q^n=C \oplus \F_q^{n-c}$. 
For any $M \in \mM_i(p,p') \cup \{
B\} $ we have $\mbox{rowsp}(\overline{M}) \subseteq \F_q^{n-c}$. It follows
$$\mC= \{ C \oplus \mbox{rowsp}(\overline{M}) : M \in \mM_i(p,p') \cup \{
B\} \}.$$
By  \cite{noi}, Theorem 13 and Proposition 17,
the set $\{ \mbox{rowsp}(\overline{M}) : M \in \mM_i(p,p') \cup \{ B\} \}$ is a
partial spread in $\mG_q(k-c,n-c)$ of cardinality
$(q^{n-c}-q^r)/(q^{k-c}-1)-q^r+1$.
The theorem now follows from Remark \ref{legame}.
\end{proof}

\begin{notation}
 We denote the sunflower of Theorem \ref{nsconstr} by $\mF_q(k,n,c,p,p')$, and
we  call it a \textbf{sunflower code}. If $h=2$, then the construction does not depend on $p$ and we denote the code by $\mF_q(k,n,c,p')$.
In the sequel we will work with a fixed integer $0 \le c \le k-1$ and with fixed polynomials $p$ and $p'$ as in Theorem \ref{nsconstr}.
\end{notation}

\begin{example}
 Let $q=2$, $c=1$, $k=3$ and $n=6$. Let $p':=x^3+x+1 \in \F_2[x]$. The companion matrix of $p'$ is
$$P'=\begin{bmatrix}
     0 & 1 & 0 \\ 0 & 0 & 1 \\ 1 & 1 & 0
    \end{bmatrix}.
$$
A codeword of  $\mF_q(3,6,1,p')$ is either the space generated by
the rows of the matrix
$$B=\begin{bmatrix}
   1 & 0 & 0 & 0 & 0 & 0 \\
   0 & 0 & 0 & 0 & 1 & 0 \\
   0 & 0 & 0 & 0 & 0 & 1 \\
  \end{bmatrix},
$$
or the space generated by the rows of a matrix of the form
$$\begin{bmatrix}
   1 & \begin{array}{cc} 0 & 0 \end{array} & \begin{array}{ccc} 0 & 0 & 0
\end{array} \\ 
\begin{array}{cc} 0 \\ 0 \end{array} & I_2 & A_{(2)}
  \end{bmatrix},
$$
where $I_2$ is the $2 \times 2$ identity matrix, and $A_{(2)}$ denotes the last
two rows of a matrix $A\in \F_2[P']$. One can easily check that
$|\mF_2(3,6,1,p,p')|=2^3+1$.
\end{example}

For most choices of the parameters, sunflower codes have asymptotically optimal 
cardinality, as the following result shows.

\begin{proposition}
Let $n \ge 3k-1$, and let $r$ denote the remainder obtained dividing $n-c$ by $k-c$. 
For $q \gg 0$ we have $$e_q(k,n,c)-|\mF_q(k,n,c,p,p')|\le q^r-1.$$
In particular, 
$$\lim_{q\to\infty}\frac{|\mF_q(k,n,c,p,p')|}{e_q(k,n,c)}=1.$$
\end{proposition}

\begin{proof}
By Proposition \ref{noi} and
Theorem \ref{bo} we have
$e_q(k,n,c)=e_q(k-c,n-c,0) \le \frac{q^{n-c}-q^r}{q^{k-c}-1}$.
By Theorem \ref{nsconstr} it follows that
$$e_q(k,n,c)-|\mF_q(k,n,c,p,p')| \le
\frac{q^{n-c}-q^r}{q^{k-c}-1}-|\mF_q(k,n,c,p,p')|=q^r-1.$$
If in addition $n\geq 3k-1$, by Proposition \ref{noi} the integers $q,k,n,c$ satisfy condition $(*)$ for $q \gg 0$.
By definition $|\mF_q(k,n,c,p,p')| \le e_q(k,n,c)$.
It follows that for $q \gg 0$
\begin{equation}\label{confronto}
e_q(k,n,c)-q^r+1 \le |\mF_q(k,n,c,p,p')| \le e_q(k,n,c).\end{equation}
Since $r \le k-c-1 \le k-1 < n-k$, the second part of the thesis follows taking the limit of (\ref{confronto}).
\end{proof}

\section{Decoding sunflowers codes} \label{de}

In this section we provide an efficient decoding algorithm for the sunflower codes that we constructed in Section~\ref{con}, by reducing decoding sunflower codes to decoding partial spread codes. 

\begin{definition}\label{drr}
 Let $1 \le t \le n$ be an integer. A matrix $M$ of size $t \times n$
over $\F_q$ is said to be in \textbf{reduced row-echelon form} if:
\begin{enumerate}
 \item $M$ is in row-echelon form;
\item the first non-zero entry of each row of $M$ is a 1, and it is the only non-zero entry
in its column.
\end{enumerate}
\end{definition}

\begin{remark} \label{rr}
 It is well-known that for any $1 \le t \le n$ and any $t$-dimensional
$\F_q$-subspace $X \subseteq \F_q^n$, 
there exists a unique $t \times n$ matrix
in reduced row-echelon form and with entries in $\F_q$ such that
$\mbox{rowsp}(M)=X$.
\end{remark}

\begin{notation}
 We denote the matrix $M$ of Remark \ref{rr} by $\mbox{\textnormal{RRE}}(X)$.
\end{notation}

The decoding algorithm for sunflower codes that we propose is based on the
following result.

\begin{theorem} \label{dec}
 Let $V \in \mF_q(k,n,c,p,p')$, $V=\mbox{rowsp}(M)$ where $M$ is as in Theorem \ref{nsconstr}:
$$M= \mbox{\textnormal{rowsp}} \begin{bmatrix}
                   I_c & 0_{c \times n-c} \\
                   0_{k-c \times c} & B
                  \end{bmatrix},$$
with $B$ of size $(k-c)\times (n-c)$. Let $X \subseteq \F_q^n$
be a subspace of dimension $1 \le t \le k$. Assume that $X$ decodes to $V$,
i.e.,
$d(V,X) < k-c$. Then:
\begin{enumerate}
 \item $t > c$ and there exist matrices $X_1, X_2, X_3$
of size $c \times c$, $c \times (n-c)$ and $(t-c) \times (n-c)$ respectively,
such that
$$\mbox{\textnormal{RRE}}(X)=  \begin{bmatrix}
                               X_1 & X_2 \\ 0_{(t-c) \times c} & X_3 
                              \end{bmatrix}.
$$
\item $d(\mbox{\textnormal{rowsp}}(B), \mbox{\textnormal{rowsp}}(X_3))<k-c$.
\end{enumerate}
\end{theorem}

\begin{proof}
The condition $d(V,X)<k-c$ is equivalent to
$\dim(V+X)< k+(t-c)/2$. In particular we have
$k=\dim(V) \le \dim(V+X)<k+(t-c)/2$, and so
$t>c$.
Notice moreover that by Definition \ref{drr} the $i$-th row of any matrix in 
reduced row-echelon form contains at least 
$i-1$ zeros. As a consequence, 
$$\mbox{RRE}(X)=\begin{bmatrix}
                              X_1 & X_2 \\ 0_{t-c \times c} & X_3 
                             \end{bmatrix}$$
for some matrices $X_1$, $X_2$ and $X_3$ of size $c \times c$, $c \times (n-c)$ and $(t-c) \times (n-c)$ respectively.  
To simplify the notation, we omit the size of the zero matrices in the sequel.
The condition $\dim(V+X)< k+(t-c)/2$
may be written as
$$\mbox{rk} \begin{bmatrix}
      I_c & 0 \\ 0 & B \\ X_1 & X_2 \\ 0 & X_3
     \end{bmatrix} < k+(t-c)/2.
$$
Hence we have
$$\mbox{rk} \begin{bmatrix} B \\ X_3 \end{bmatrix} =
\mbox{rk} \begin{bmatrix} I_c & 0 \\ 0 & B \\ 0 & X_3 \end{bmatrix} -c \le
\mbox{rk} \begin{bmatrix} I_c & 0 \\ 0 & B \\ X_1 & X_2 \\ 0 & X_3 \end{bmatrix} -c <
(k-c)+(t-c)/2.
$$
Since $\dim(X)=t$, we have $\mbox{rk}(X_3) = t-c$. It follows that
\begin{eqnarray*}
d(\mbox{rowsp}(B), \mbox{rowsp}(X_3)) &=& 2 \mbox{rk}
\begin{bmatrix} B \\ X_3 \end{bmatrix}-\mbox{rk}(B)-\mbox{rk}(X_3) \\
&<& 2(k-c)+t-c-(k-c)-(t-c) \\
&=& k-c,
\end{eqnarray*}
as claimed.
\end{proof}

Theorem \ref{dec} provides in particular 
the follwing efficient algorithm to decode a sunflower code.

\begin{algo}[Decoding a $\mF_q(k,n,c,p,p')$ code] \label{algor}
 \mbox{ }
 \begin{itemize}
  \item  \textbf{Input:} A decodable subspace $X \subseteq \F_q^n$ of
dimension $t \le k$.
\item  \textbf{Output:} The unique $V \in \mF_q(k,n,c,p,p')$ such 
that $d(V,X)<k-c$, given as a matrix in row-reduced echelon form whose rowspace
is $V$.
 \end{itemize}
\begin{enumerate}
 \item Compute $M:=\mbox{\textnormal{RRE}}(X)$.
\item Delete from $M$ the first $c$ rows and columns, obtaining
a matrix $\overline{M}$ of size $k-c \times n-c$.
\item Apply partial spread decoding to $\mbox{\textnormal{rowsp}}(\overline{M})$
as described in \cite{noi}, Section 5, and obtain a matrix $N$ of size $k-c
\times n-c$.
\item The result is $V= \mbox{\textnormal{rowsp}}\begin{bmatrix}
                      I_c & 0 \\ 0 & N
                     \end{bmatrix}$.
\end{enumerate}
\end{algo}

\begin{remark}
For any decodable subspace, $t>c$ by Theorem \ref{dec}.
The assumption $t \le k$ is not restrictive from the following point of view: The
receiver may collect incoming vectors until the received subspace has
dimension $k$, and then attempt to decode the collected data.
We also notice that the computation of $\mbox{RRE}(X)$ has a low computational
cost. Indeed, the receiver obtains the subspace $V$ as
the span of incoming vectors, i.e., as the rowspace of a 
matrix. The reduced row-echelon form of such matrix may be computed
by Gaussian elimination.
\end{remark}

\section{The orthogonal of a sunflower code} \label{duals}

By Proposition \ref{full} and Lemma \ref{not}, the orthogonals of sunflower
codes of Theorem \ref{nsconstr} are equidistant codes
that are not sunflowers. Moreover, they are asymptotically optimal equidistant
codes for sufficiently large parameters (Remark \ref{dua} and Theorem \ref{class}). 
We can easily write them as rowspaces of 
matrices, as we show
in this section. We will need the following preliminary lemma, whose proof is left to the reader.

\begin{lemma}\label{cons}
 Let $N$ be a $t \times (n-t)$ matrix over $\F_q$. We have
$$\mbox{rowsp}\left( \begin{bmatrix}I_t & N \end{bmatrix} \right)^\perp = 
\mbox{rowsp}\left( \begin{bmatrix}-N^t & I_{(n-t)\times (n-t)} \end{bmatrix}
\right).$$
\end{lemma}

\begin{remark} \label{trickdual}
 Lemma \ref{cons} allows us to construct the orthogonal of a vector space $V$
 given as the rowspace of a full-rank matrix $M$ in reduced row-echelon form.
Indeed, if $M$ is such a matrix of size, say, $t \times n$, then there exists a
permutation $\pi:\{1,...,n \} \to 
\{1,...,n\}$
such that $M^\pi$ has the form $\begin{bmatrix}I_t & N \end{bmatrix}$, where
$M^\pi$ is the matrix whose $\pi(i)$-th columns is the $i$-th column of $M$. By
Lemma \ref{cons}
we have $V^\perp=\mbox{rowsp}\left( \begin{bmatrix}-N^t & I_{(n-t)\times (n-t)}
\end{bmatrix}^{\pi^{-1}} \right)$.
 \end{remark}
 
\begin{remark}
Remark \ref{trickdual} allows us to describe in matrix form the orthogonal
of a sunflower code $\mC=\mF_q(k,n,c,p,p')$. Indeed, following the notation of Theorem
\ref{nsconstr}, the orthogonal of the rowspace of the matrix
$$\begin{bmatrix}
I_c & 0_{c\times (k-c)} & \cdots & \cdots & \cdots & \cdots & \cdots & 0_{c\times (k-c)} & 0_{c\times (k-c+r)} \\
0_{(k-c) \times c} & 0_{k-c} & \cdots & 0_{k-c} & I_{k-c} & A_{i+1} & \cdots & A_{h-1} & A_{[k-c]}
\end{bmatrix}$$
is the rowspace of the matrix
$$\begin{bmatrix}

0_{(k-c) \times c} &  \multicolumn{3}{c}{\multirow{3}{*}{$I_{(i-1)(k-c)}$}} & 0_{k-c} &
 \cdots & \cdots  & \cdots & \cdots \\

\vdots &  &  & & \vdots & \vdots & \vdots & \vdots & \vdots \\

\vdots & &  & & 0_{k-c} & \vdots & \vdots & \vdots & \vdots \\

\vdots &  0_{k-c} & \cdots & 0_{k-c} & -A_{i+1}^t &
\multicolumn{4}{c}{\multirow{4}{*}{$I_{n-k-(i-1)(k-c)}$}}      \\

\vdots &  \vdots & \vdots & \vdots & \vdots & &  & &  \\

\vdots &  \vdots & \vdots & \vdots & -A_{h-1}^t &   &    & &  \\

0_{(k-c+r)\times c} &  0_{(k-c+r)\times (k-c)} & \cdots & \cdots & -A_{[k-c]}^t &   &    & &  \\

\end{bmatrix}.$$
%This fact can also be directly checked by looking at the two matrices.
\end{remark}

Algorithm \ref{algor} and Remark \ref{trickdual} can also be combined to
efficiently decode
the orthogonal of a sunflower code. 

\begin{remark}
Let $\mC=\mF_q(k,n,c,p,p')$ be a sunflower code, and let $X
\subseteq \F_q^k$
be a received $t$-dimensional space. Since $d(\mC)=d(\mC^\perp)$ and
$d(X,V^\perp)=d(X^\perp,V)$ for all $V \in \mC$, the space $X$ decodes to
$V^\perp$ in $\mC^\perp$ if and only if $X^\perp$ decodes to $V$ in $\mC$.
This gives the following Algorithm \ref{algor2} to decode the orthogonal of a sunflower code.
\end{remark}

\begin{algo}[Decoding a $\mF_q(k,n,c,p,p')^\perp$ code] \label{algor2}
 \mbox{ }
 \begin{itemize}
  \item  \textbf{Input:} A decodable subspace $X \subseteq \F_q^n$ of
dimension $t \ge n-k$.
\item  \textbf{Output:} The unique $V \in \mF_q(k,n,c,p,p')$ such 
that $d(V^\perp,X)<k-c$, given as a matrix  whose rowspace
is $V$.
 \end{itemize}
\begin{enumerate}
 \item Compute $L:=\mbox{\textnormal{RRE}}(X)$.
\item Use Remark \ref{trickdual} to construct a matrix $L'$ such that
$\mbox{rowsp}(L')=X^\perp$.
\item  Compute the reduced row-echelon form, say $M$, of $L'$. Since
$t \ge n-k$, $M$ will have at most $k$ rows, as required by Algorithm
\ref{algor}.
\item Delete from $M$ the first $c$ rows and columns, obtaining
a matrix $\overline{M}$ of size $(k-c) \times (n-c)$.
\item Apply partial spread decoding to $\mbox{\textnormal{rowsp}}(\overline{M})$
as described in \cite{noi}, Section 5, and obtain a matrix $N$ of size $k-c
\times n-c$.
\item We have $V^\perp= \mbox{\textnormal{rowsp}}\begin{bmatrix}
                      I_c & 0 \\ 0 & N
                     \end{bmatrix}$. Use Remark \ref{trickdual} to describe $V$
as the
rowspace of a matrix.
\end{enumerate}
\end{algo}

%
%
%
%\section*{Conslusions}
%In this paper we provide a classification
%of equidistant subspace codes of maximum cardinality for sufficiently
%large paramaters, proving that such codes are either sunflowers, or the orthogonals of
%sunflowers. Moreover, we study some equidistant codes with special parameters, and establish
%general properties of equidistant codes that are not sunflowers. We also propose a systematic 
%construction of sunflower codes whose size is asymptotically optimal, and provide an efficient 
%decoding algorithm for them and for their orthogonal codes.

\end{document}